\documentclass[11pt,draftclsnofoot,peerreviewca,onecolumn,letterpaper]{IEEEtran}

\advance \topmargin by -\headheight
\advance \topmargin by -\headsep
\evensidemargin \oddsidemargin
\usepackage{amsmath}
\usepackage{amssymb,theorem}
\usepackage[pdftex]{graphics}
\usepackage{epsfig}
\usepackage[pdftex]{color}
\usepackage{psfrag,bbm}
\usepackage{relsize}

\usepackage{latexsym}
\usepackage{epsfig}

\usepackage{graphicx}
\usepackage[font=small]{caption}
\usepackage[font=footnotesize]{subcaption}
\usepackage{wrapfig}
\usepackage{cite}

\long\def\comment#1{}

\newfont{\bbb}{msbm10 scaled 800}

\newfont{\bb}{msbm10 scaled 1100}
\newcommand{\CC}{\mbox{\bb C}}

\newcommand{\EE}{\mbox{\bb E}}

\newcommand{\sv}{{\bf s}}

\newcommand{\xv}{{\bf x}}
\newcommand{\yv}{{\bf y}}
\newcommand{\zv}{{\bf z}}

\newcommand{\Cm}{{\bf C}}

\newcommand{\Hm}{{\bf H}}

\newcommand{\Sm}{{\bf S}}

\newcommand{\Um}{{\bf U}}

\newcommand{\Vm}{{\bf V}}
\newcommand{\Xm}{{\bf X}}
\newcommand{\Ym}{{\bf Y}}
\newcommand{\Zm}{{\bf Z}}

\newtheorem{thm}{Theorem}
\newtheorem{lemma}{Lemma}
\newtheorem{remark}{\indent \bf Remark}

\begin{document}

\title{On Uplink-Downlink Duality for Cellular IA}

\author{Vasilis~Ntranos$^\dagger$, 
        Mohammad~Ali~Maddah-Ali$^\ast$,  and
        Giuseppe~Caire$^\dagger$ \\
        $^\dagger$University of Southern California, Los Angeles, CA,  USA\\ 
				$^\ast$Bell Labs, Alcatel-Lucent, Holmdel, NJ, USA }
        
\maketitle

\begin{abstract}
In our previous work \cite{nmc14,nmc14b} we considered the \emph{uplink} of a hexagonal cellular network topology and showed that linear ``one-shot'' interference alignment (IA) schemes  are able to achieve the optimal degrees of freedom (DoFs) per user, under a  decoded-message passing framework that allows base-stations to exchange their own decoded messages over local backhaul links. In this work, we provide the dual framework for the \emph{downlink} of cellular networks with the same backhaul architecture, and show that for every ``one-shot'' IA scheme that can achieve $d$ DoFs per user in the uplink, there exists a dual ``one-shot'' IA scheme that can achieve the same DoFs in the downlink. 
To enable ``Cellular IA'' for the downlink, base-stations will now use the same local backhaul links to exchange quantized versions of the 
dirty-paper precoded signals instead of user messages.
\end{abstract}
\vspace{-0.2in}
\section{Introduction}
In \cite{nmc14,nmc14b} we have shown that practical ``one-shot'' interference alignment schemes can achieve the optimal degrees of freedom in the uplink of a cellular network topology in which base-stations (receivers) can exchange decoded messages locally over the backhaul links. A natural question that comes to mind is whether similar results can be obtained for the downlink of such networks. In the downlink, the local backhaul connections between base-stations can be used to enable \emph{transmitter cooperation}, as opposed to receiver cooperation in the uplink.

Cooperation of multiple base-stations in the downlink with some form of data or signal sharing through the wired backhaul network is a subject of intense research both in terms of information theoretic fundamental limits and in terms of practical signaling/coding schemes. 
Most works have focused on the setting where a central processor is connected to all the cooperating base-stations through orthogonal noiseless links of given capacity 
$R_0$ (e.g., see \cite{ssps09b,simeone09,hongcaire13}).
This model is a special case of the general broadcast-relay channel \cite{krv12}
where a transmitter wishes to send independent messages to multiple users through a layer of relays (the base-stations) and where the first hop of this two-hop communication scenario is formed by the set of orthogonal noiseless links. From a practical viewpoint, such architecture is usually referred to as ``cloud base-station'' (or C-RAN, in the parlance of 3GPP-LTE). The first hop (from the central processor to the base-stations is often referred to as ``fronthaul'', and the base-stations are just simple antenna heads, whose only task consists of converting the precoded digital signal generated by the central processor into an RF signal to be dirty-paper transmitted on the downlink \cite{lightradio, flanagan11,IBMcloud10}.

A different approach is taken in \cite{lsw07,llsw12} where 
local message sharing at the base-stations is considered for the linear Wyner model, 
which induces a cognitive interference channel with cognition at the transmitter, where  cognition corresponds to message side information according to the fixed sharing pattern.
Also, in  \cite{gav12,UIUC11,UIUC12,UIUC14} such setting is extended to the case where base-station
an average backhaul rate constraint is imposed, which allows the time sharing between different message sharing patterns and overall the achievability of tight DoFs results for 
the linear Wyner model and for certain two-dimensional models with limited connectivity of the interference  graph. 

The scheme presented in this paper is radically different from all of the above. On one hand, it requires only
local communication with per-link rate constraint between neighboring base-stations, without joint central processing, and allows to handle two dimensional hexagonal cell patterns (with or without sectors), with higher connectivity than what studied in \cite{gav12,UIUC11,UIUC12,UIUC14}. 
On the other hand, the local communication between base-stations is used to exchange a suitably quantized version of the
dirty-paper coded (DPC) signal that each base-station can use in order to precode for the known interference. 
This induces a directed interference graph for the downlink, and allows alignment techniques that can be regarded as ``duals'' of what we have used for the uplink in [1,2].

More specifically, we will propose here  a \emph{successive encoding scheme} for the downlink, based on dirty-paper coding (DPC), that enables a directed network interference cancellation over the backhaul across the entire cellular system; base-stations will  first quantize and then share their dirty-paper precoded signals with their neighbors, who can in turn successively encode their messages using DPC to avoid  the known interference. Within this framework, we will show that any  DoFs that are achievable by ``one-shot'' interference alignment in the uplink of a cellular system with a given decoding order $\pi$, are also achievable in the downlink of this network with the same linear IA precoding scheme, as long as the corresponding \emph{encoding order} $\overline \pi$ (under which base-stations  encode, quantize and share their dirty-paper signals) is reversed.

 It is worth pointing out that dirty-paper coding  plays a fundamental role in the proposed achievable scheme for the downlink (unlike in \cite{gav12,UIUC11,UIUC12,UIUC14}, where optimal DoFs
can be achieved with linear precoding). Furthermore, it is also remarkable that, under our framework, base-stations do not share neither messages nor quantized received signals (as in the Quantize-Remap and Forward paradigm of \cite{adt11}), but quantized (dirty-paper) precoded signals.

 This paper is organized as follows. In Section \ref{sec:downlink} we will describe  the  successive DPC scheme that enables  the corresponding network interference cancellation framework, and in Section \ref{sec:duality} we will show that under this framework, ``one-shot'' schemes based on cellular interference alignment for the downlink can be directly obtained by uplink-downlink duality.

\section{Successive encoding for the downlink}\label{sec:downlink}

In this section we will focus on two neighboring base-stations of the cellular network that are connected through a limited capacity backhaul link and describe how they can successively encode their messages using dirty-paper coding such that interference is pre-canceled in one direction. 
We will consider here for simplicity the case where both transmitters and receivers are equipped with a single antenna  in order to outline the main idea behind our successive encoding scheme.

\begin{figure}[ht]

                \centering
                \includegraphics[width=\columnwidth]{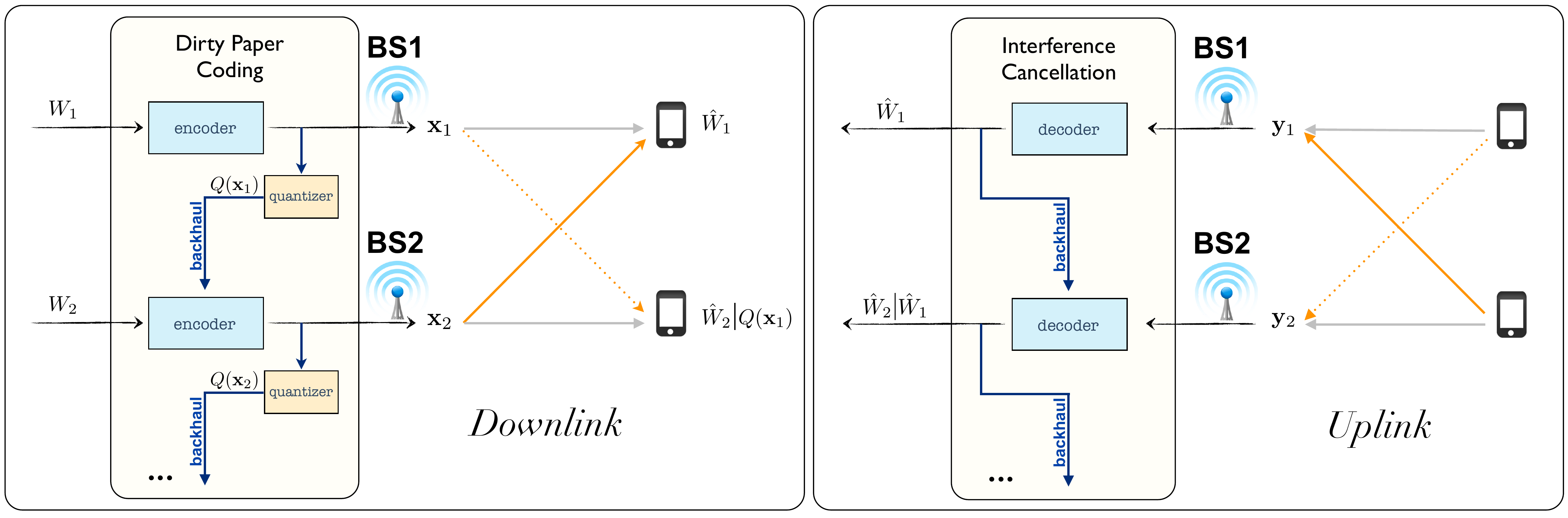}
                \caption{Successive decoding in the uplink versus successive encoding in the downlink. In both cases, base-station $1$ (BS1) will use the backhaul to give the corresponding information to base-station $2$ (BS2). In the uplink, BS2 can use $\hat{W}_{1}$ to reconstruct the corresponding signal and subtract the interference coming from user $1$. In the downlink, BS2 can use $Q(\xv_{1})$ and dirty-paper coding  (DPC) to avoid  interference from BS1.}
                \label{uldl}

\end{figure}

\clearpage
The received signals observed at the mobile users associated with base-station 1 (BS1) and base-station 2 (BS2) are given by
\begin{align}
\yv_{1}&= \tilde h_{11}\xv_{1} + \tilde h_{12}\xv_{2} + \zv_{1}, \nonumber \\
\yv_{2}&= \tilde h_{22}\xv_{2} +\tilde h_{21}\xv_{1} + \zv_{2}, 
\end{align}
where $\xv_{1}$, $\xv_{2}$  are the transmitted signals (represented as vectors with block length $n$) of BS1 and BS2, satisfying the average power constraint $\frac{1}{n}\EE[\xv_{i}^{\rm H}\xv_{i}] \leq P$, $i=1,2$ and $\zv_{i}$ is i.i.d Gaussian noise with unit variance. We will assume without loss of generality that the BS1 has already encoded its message using DPC, to eliminate some other interfering links in the network, and focus on BS2. 
Fig.~1 shows the successive encoding scheme for the downlink (in the setting we consider here) in comparison to the successive decoding scheme that we used in our previous work for the uplink. 

In the downlink, BS1 will first quantize its transmitted dirty-paper signal $\xv_{1}$ to obtain $Q(\xv_{1})$. {Since $\xv_{1}$ is Gaussian i.i.d \footnote{It is well-known that the dirty paper precoded signal can be treated as Gaussian iid. This follows from the fact that the random variable X 
forming the auxiliary  variable $U = X + \alpha S$ in Costa's coding scheme \cite{costaDPC} is Gaussian and independent of the interference S, and from standard strong typicality arguments (see e.g., the appendix of \cite{wnc10}).
Also, if the universal lattice precoding scheme of \cite{zse02}
is used instead of Costa's scheme, it is well-known that the dithered modulo lattice precoded signal
is Gaussian i.i.d. in the limit of large dimension for a sequence of shaping-good lattices. 
}} with average power $P$, its quantized version $Q(\xv_{1})$ will also be Gaussian i.i.d, and can be represented at rate $R(D)=\log(P/D)$ with average distortion (quantization noise variance) given by $\frac{1}{n}\EE\left[||\xv_{1} - Q(\xv_{1})||^{2}\right] \leq  D$.  In order to keep the quantization noise at the system's noise level, we will set the distortion $D = 1$. 
 Now, assuming that the backhaul rate between BS1 and BS2 is at least $\log(P)$, we can  let BS1 give its quantized dirty-paper signal $Q(\xv_{1})$ to BS2. As a result, BS2 will know the quantized part of interference $Q(\xv_{1})$ coming from BS1 and can use it to successively encode its own dirty-paper signal as follows.  The observed signal at the intended receiver of BS2 can be written as 
\begin{align}
\yv_{2} &= \tilde h_{22}\xv_{2} +\tilde h_{21}\xv_{1} + \zv_{2} \nonumber\\
&= \tilde h_{22}\xv_{2} +\tilde h_{21}Q(\xv_{1}) + \tilde h_{21}\underbrace{(\xv_{1}-Q(\xv_{1}))}_{{\mbox{\footnotesize quantization noise}}} + \zv_{2} \nonumber\\ &= \tilde h_{22}\xv_{2} +\tilde h_{21}Q(\xv_{1}) +{\zv_{Q}}+ \zv_{2},
\end{align}
where  ${\zv_{Q}}\triangleq \tilde h_{21}(\xv_{1}-Q(\xv_{1}))$ denotes the effective i.i.d Gaussian noise with variance $|\tilde h_{21}|^{2}$ due to quantization. Notice that since the quantization noise  ${\zv_{Q}}$ is independent of $\xv_{1}$, the above observation can be written in the standard form:    
\begin{equation}
\yv_{2} = \tilde h_{22}\xv_{2} + \sv + \tilde \zv,
\end{equation}
where $\sv=\tilde h_{21}\cdot Q(\xv_{1})$ is the known interference at BS2 and $\tilde \zv = {\zv_{Q}}+ \zv_{2}$ is the effective  Gaussian noise with variance $1+|\tilde h_{21}|^{2}$. Using DPC at BS2 to avoid the known  interference $\sv$, we can  obtain an achievable rate at user 2 given by
\vspace{-0.1in}
\begin{equation}
R_{2} = \log\left(1+ \frac{|\tilde h_{22}|^{2}P}{1+|\tilde h_{21}|^{2}}\right),
\end{equation}
which has the same pre-log factor equal to 1 DoF as if interference was not present, due to the fact that the quantization noise variance is constant and not a function of $P$.
Therefore, at high SNR, we can see that this successive encoding scheme for the downlink has exactly the same network interference cancellation properties as the successive decoded message passing scheme that we have used for the uplink; In the following section we will use this scheme to obtain the corresponding uplink-downlink duality result for the ``one-shot'' DoFs achievable by Cellular IA.


\begin{remark} \rm It is worth pointing out that sharing the quantized dirty-paper signals is fundamental for this scheme to be embedded in the context of a larger cellular network. To enable network interference cancellation, one could be tempted to use an approach in which base-stations share \emph{user messages} instead of quantized codewords  over the backhaul.  
However, we can see that in that case, interference would  propagate through the cellular system -- from neighbor-to-neighbor, along the network interference cancellation paths -- and subsequent base-stations would observe interfering signals that are functions of all their predecessors' messages in the encoding order.
\end{remark}

\section{Uplink-Downlink Duality}\label{sec:duality}

In this section we will consider the cellular model introduced in \cite{nmc14,nmc14b}. For completeness, we  provide here the main definitions that will be used in the rest of this section.

\begin{itemize}
\vspace{-0.1in}
\item
{\bf Interference Graph:}
The interference graph $\cal G (\cal V, \cal E)$ of a cellular network is an undirected graph in which vertices $v\in\cal V$ represent transmit-receive pairs within a cell and edges $(u,v) \in \cal E$ indicate interfering neighbors.
\hfill $\lozenge$

\item 
{\bf Decoding/Encoding order:}
A decoding (or encoding) order $\pi$ is defined as a 
partial order ``$\prec_{\pi}$'' over the  set of vertices $\cal V$ in the above interference graph.  According to $\pi$, the cell associated with vertex $v \in \cal V$ will decode (or encode) its message before the one associated with vertex $u \in \cal V$ if $v \prec_{\pi}u$. 
\hfill $\lozenge$

\item
{\bf Directed Interference Graph ${\cal G}_{\pi}$:} For a given partial order ``$\prec_{\pi}$'' on $\cal V$, the directed interference graph is defined as ${\cal G}_{\pi}({\cal V}, {\cal E}_{\pi})$ 
where 
${\cal E}_{\pi}$ is a set of ordered pairs $[u,v]$ given by 
$
{\cal E}_{\pi} = \left\{[u,v] : (u,v) \in {\cal E} \mbox{ and }  v\prec_{\pi}u   \right\}
$. \hfill $\lozenge$

\end{itemize}

A first step towards our main uplink-downlink DoF duality result will be to show that the encoding scheme based on dirty-paper coding that we introduced in the previous section is indeed able to successively remove directed interfering links in the downlink, across the entire network $\cal G (\cal V, \cal E)$, according to a given (predefined) encoding order $\pi$.

 Let $\Vm_{v}\in \CC^{M\times d_{v}},\Um_{v}\in \CC^{N\times d_{v}}$ denote the transmit and receive beamforming matrices associated with each cell $v\in \cal V$, where $M$, $N$ is the number of the available transmit/receive antennas and $d_{v}$ is the number of transmitted signals (of block length $n$) in cell $v\in \cal V$  denoted by $\Xm_{v} \in \CC^{d_{v} \times n}$. 
Further, let $\Hm_{uv} \in \CC^{N\times M}$ denote the (constant, flat-fading) channel gains  between the transmitter of cell $v\in \cal V$ and the receiver of the cell  $u\in \cal V$, that are chosen at random from a continuous distribution and are identically zero for all $(u,v)\notin {\cal E}, u\neq v$.

\begin{lemma}
The effective channel between the downlink transmit-receive pair  associated with cell $u\in \cal V$ after DPC is given by
\begin{equation}
\Um_{u}^{\rm H}\Ym_{u} = \Um_{u}^{\rm H}\Hm_{uu}\Vm_{u}\Xm_{u} + \sum_{v: \;u \prec_{\pi} v }  \Um_{u}^{\rm H}\Hm_{uv}\Vm_{v}\Xm_{v} + \Um_{u}^{\rm H}\tilde\Zm_{u},
\label{eff}
\end{equation}
where the columns of $\tilde\Zm_{u}\in \CC^{ N \times n}$ are i.i.d zero-mean Gaussian noise vectors with covariance  
\begin{equation}
\Cm_{\zv} \triangleq \frac{1}{n}\EE[\tilde\Zm_{u}^{\rm H}\tilde\Zm_{u}]\; =\; {\bf I}\;\, + \sum_{v: \; u \phantom{.}_{\pi}\succ v} \Hm_{uv}^{\phantom{H}}\Vm_{v}^{\phantom{H}}\Vm_{v}^{\rm H}\Hm_{uv}^{\rm H}.
\label{cov}
\end{equation}
\end{lemma}
\begin{proof}
The projected received signal at cell $v\in \cal V$ is given by 
$$\Um_{u}^{\rm H}\Ym_{u} = \Um_{u}^{\rm H}\Hm_{uu}\Vm_{u}\Xm_{u} + \sum_{v\in {\cal N}({u})}  \Um_{u}^{\rm H}\Hm_{uv}\Vm_{v}\Xm_{v} + \Um_{u}^{\rm H}\Zm_{u},$$
where ${\cal N}({u})$ denotes the set of all the interfering neighbors of  $u\in \cal V$. This set can be partitioned according to the encoding order $\pi$ into two sets, $\{v \in {\cal N}({u}) : \;u \prec_{\pi} v \}$ and $\{v \in {\cal N}({u}) : \; u \phantom{.}_{\pi}\hspace{-0.05in}\succ v \}$. All the transmitters $v$ that belong to the set that has already encoded their messages (i.e, $u \phantom{.}_{\pi}\hspace{-0.05in}\succ v$) will quantize their DPC signals $\Xm_{v}$ into $Q(\Xm_{v})$ at rate $d_{v}\cdot\log(P)$ and  give them to base-station~$u$ through the backhaul. The received signal can hence be written as
\begin{align}
\Um_{u}^{\rm H}\Ym_{u} \;\;=\;\; &\Um_{u}^{\rm H}\Hm_{uu}\Vm_{u}\Xm_{u} 
+ \sum_{v: \;u \prec_{\pi} v }  \Um_{u}^{\rm H}\Hm_{uv}\Vm_{v}\Xm_{v} 
+ \underbrace{\sum_{v: \;u \phantom{.}_{\pi}\succ v}  \Um_{u}^{\rm H}\Hm_{uv}\Vm_{v}\cdot Q(\Xm_{v})}_{\mbox{\footnotesize known interference}} \nonumber\\
&+ \Um_{u}^{\rm H}\underbrace{\sum_{v: \;u \phantom{.}_{\pi}\succ v}  \Hm_{uv}\Vm_{v}(\Xm_{v}-Q(\Xm_{v}))}_{\mbox{\footnotesize quantization noise}}
+ \Um_{u}^{\rm H}\Zm_{u}.
\label{7}
\end{align}
If we let $$\tilde \Zm_{u} \triangleq  \Zm_{u} + \sum_{v: \;u \phantom{.}_{\pi}\succ v}  \Hm_{uv}\Vm_{v}(\Xm_{v}-Q(\Xm_{v}))$$ and encode $\Xm_{u}$ using DPC  to avoid the known interference we obtain  the effective channel given by (\ref{eff}) with (column-wise) i.i.d Gaussian noise whose covariance is given by (\ref{cov}). 
\end{proof}

\begin{remark}{\rm Although in  the above scheme base-stations share quantized DPC signals instead of messages, the rate required for the backhaul links
in the downlink is the same (in the leading order of $P$) as the rate required for corresponding the local message-passing scheme  in the uplink. This follows from the fact that DPC in this setting is used on top of  the linear precoding scheme over the antennas: If we let $\tilde\Hm_{uv}\triangleq\Um_{u}^{\rm H}\Hm_{uv}\Vm_{v}\in \CC^{d_{u}\times d_{v}}$, $\forall (u,v)\in \cal E$, then from each encoder's perspective, DPC is performed over a $d_{u}\times d_{u}$ equivalent MIMO channel: 
$$\tilde\Ym_{u} = \tilde\Hm_{uu}\Xm_{u} + \tilde \Sm + \tilde\Zm_{u},$$
with $d_{u}$-dimensional interference $\tilde\Sm=\sum_{v: \;u \phantom{.}_{\pi}\succ v}  \tilde\Hm_{uv}\cdot Q(\Xm_{v}).$ The interference $\tilde\Sm$ is therefore known to the encoder, as long as the corresponding base-station is able to get the $d_{v}$-dimensional $Q(\Xm_{v})$ at rate $\log(P)$ per dimension over the local backhaul links. This is exactly the same backhaul rate scaling required for exchanging messages and hence both the downlink and the uplink schemes can operate under the same backhaul network infrastructure.
}  
\end{remark}

\begin{figure}[ht]

                \centering
                \includegraphics[width=\columnwidth]{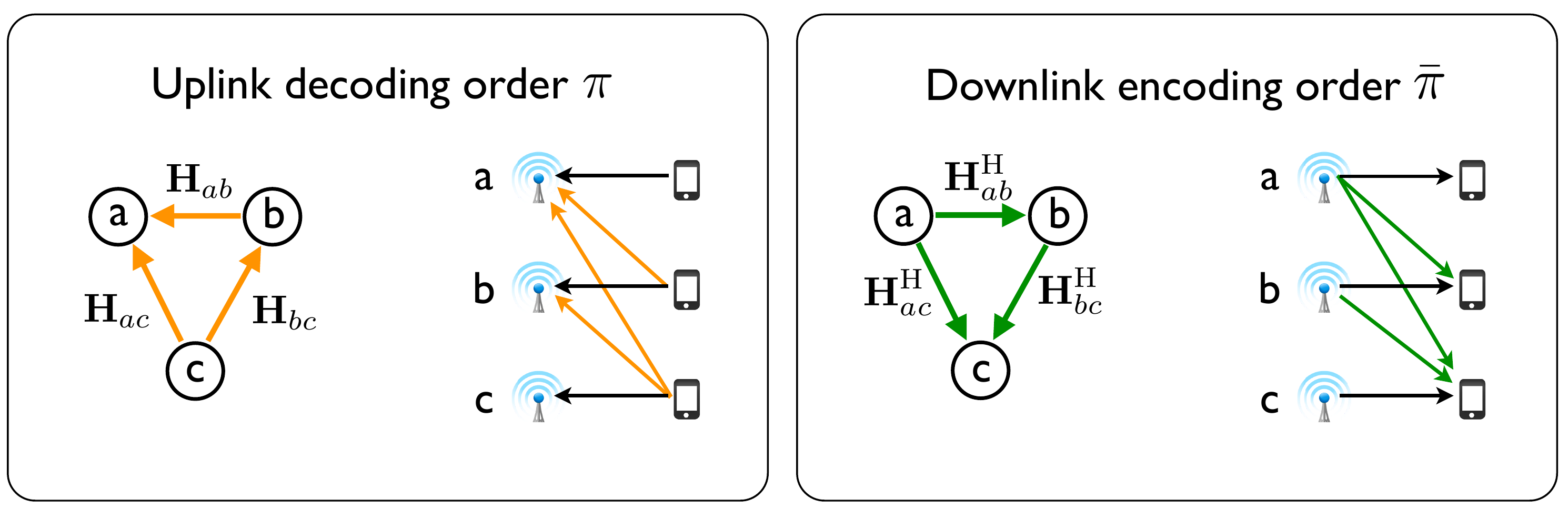}
                \caption{Uplink and Downlink with reverse decoding/encoding orders, $\pi$ and $\overline \pi$. After the corresponding network-wide interference cancellation in both cases, the remaining interference channel gains for the downlink are reciprocal to the ones obtained in the uplink, and are given by $\overline{\bf H}_{uv} = {\bf H}_{vu}^{\rm H}$, $\forall [v,u]\in {\cal E}_{\overline\pi}$.}
                \label{dual}

\end{figure}

\begin{thm}[Uplink-Downlink Duality]
Any degrees of freedom  $\{d_{v},v\in \cal V\}$ that are achievable in the uplink of the cellular network $\cal G(V,E)$ in ``one-shot'' by linear IA with beamforming matrices $\{\Vm_{v}\in\CC^{M\times d_{v}}, v\in \cal V\}$ and $\{\Um_{u}\in \CC^{N\times d_{v}}, u\in \cal V\}$,  under the network interference cancellation framework with decoding order $\pi$, are also achievable in the downlink of the same cellular network $\cal G(V,E)$ under the successive dirty-paper coding framework with the reverse encoding order $\overline\pi$ and beamforming matrices given by 
$\{\overline\Vm_{v}=\Um_{v}\in\CC^{N\times d_{v}},v\in \cal V\}$ and $\{\overline\Um_{u}=\Vm_{u}\in\CC^{M\times d_{u}},u\in \cal V\}$.
\end{thm}
\begin{proof}
Let the partial order ``$\prec_{\overline\pi}$'', defined on the set $\cal V$, be the inverse of ``$\prec_{\pi}$'' such that 
\begin{equation}
u\prec_{\overline\pi}v \Leftrightarrow v\prec_{\pi}u, \;\forall u,v \in \cal V,
\label{ord}
\end{equation}
and consider the corresponding directed interference graphs ${\cal G}_{\pi}({\cal V}, {\cal E}_{\pi})$ for the uplink and ${\cal G}_{\overline\pi}({\cal V}, {\cal E}_{\overline\pi})$ for the downlink. 
Since the degrees of freedom $\{d_{v},v\in \cal V\}$ are achievable in the uplink we can argue that the corresponding beamforming matrices chosen for the uplink, $\{\Vm_{v}\in\CC^{M\times d_{v}}, v\in \cal V\}$ and $\{\Um_{u}\in \CC^{N\times d_{v}}, u\in \cal V\}$, satisfy:
\begin{align}
&\Um_{u}^{\rm H}\Hm_{uv}\Vm_{v} = 0,\; \forall [v,u] \in \cal E_{\pi}\label{eq:condAAA}, \;\mbox{and}\\
&\mbox{rank}\left(\Um_{v}^{\rm H}\Hm_{vv}\Vm_{v}\right) = d_{v},\; \forall v\in \cal V.
\label{eq:condBBB}\end{align}

As illustrated in Fig.~\ref{dual}, for every directed edge $[v,u]\in \cal E_{\pi}$ there exists a directed edge $[u,v]\in \cal E_{\overline\pi}$ and 
the corresponding channels  are reciprocal to each other. That is,  
 the downlink channel matrices denoted by $\overline\Hm_{vu}\in \CC^{M\times N}$, $[u,v]\in \cal E_{\overline\pi}$ are given by \begin{equation}\overline\Hm_{vu}= \Hm_{uv}^{\rm H},\label{11}\end{equation} where $\Hm_{uv}\in \CC^{N\times M}$ are the corresponding uplink channel matrices associated with opposite edges $[v,u]\in \cal E_{\pi}$.
Now, we can rewrite (\ref{eq:condAAA}) as follows. 
\begin{align}
\Um_{u}^{\rm H}\Hm_{uv}\Vm_{v} = 0,\; \forall [u,v] \in \cal E_{\pi} \Leftrightarrow\;
&\Um_{v}^{\rm H}\Hm_{uv}\Vm_{u} = 0,\; \forall [u,v] \in \cal E_{\overline\pi} \label{12}\\
\Leftrightarrow\;&\Vm_{u}^{\rm H}\Hm_{uv}^{\rm H}\Um_{v} = 0,\; \forall [u,v] \in \cal E_{\overline\pi}\label{13}\\\Leftrightarrow\;
&\Vm_{u}^{\rm H}\overline\Hm_{vu}\Um_{v} = 0,\; \forall [u,v] \in \cal E_{\overline\pi},\label{14}
\end{align}
where (\ref{12}) follows from the fact that $\pi$ and $\overline\pi$ satisfy (\ref{ord}), (\ref{12}) is obtained  by transposing all equations, and (\ref{14}) by substituting the downlink channel matrices from (\ref{11}).

It has become clear now from (\ref{14}) that if we choose the downlink transmit beamforming matrices $\overline\Vm_{v} \in \CC^{N\times d_{v}}$ to be the corresponding uplink receive beamforming matrices $\Um_{v} \in \CC^{N\times d_{v}}$ and vice versa (i.e., $\overline\Um_{u}=\Vm_{u}$), the following IA conditions are satisfied in the downlink: 
\begin{align}
&\overline\Um_{u}^{\rm H}\overline\Hm_{uv}\overline\Vm_{v} = 0,\; \forall [v,u] \in \cal E_{\overline\pi}, \label{cond}\;\mbox{and}\\
&\mbox{rank}\left(\overline\Um_{v}^{\rm H}\overline\Hm_{vv}\overline\Vm_{v}\right) = d_{v},\; \forall v\in \cal V.\label{z}
\end{align}

Now from Lemma~1 we have that the  signal observation for every receiver $u\in \cal V$ is given by 
\begin{equation}
\overline\Um_{u}^{\rm H}\Ym_{u} = \overline\Um_{u}^{\rm H}\overline\Hm_{uu}\overline\Vm_{u}\Xm_{u} + \sum_{v: \;u \prec_{\pi} v }  \overline\Um_{u}^{\rm H}\overline\Hm_{uv}\overline\Vm_{v}\Xm_{v} + \overline\Um_{u}^{\rm H}\tilde\Zm_{u},
\end{equation}
and from (\ref{cond}) we can see that $\sum_{v: \;u \prec_{\pi} v }  \overline\Um_{u}^{\rm H}\overline\Hm_{uv}\overline\Vm_{v}\Xm_{v} = 0$, which in turn yields
\begin{equation}
\overline\Um_{u}^{\rm H}\Ym_{u} = \overline\Um_{u}^{\rm H}\overline\Hm_{uu}\overline\Vm_{u}\Xm_{u} + \overline\Um_{u}^{\rm H}\tilde\Zm_{u}.
\end{equation}

From (\ref{z}) and since the noise variance does not scale with the transmit power $P$, we can argue that every transmit-receive pair  $u\in \cal V$ in the downlink cellular network $\cal G(V,E)$, will achieve $d_{u}$ degrees of freedom and we conclude the proof.
\end{proof}

%
%

\bibliographystyle{ieeetr}
\bibliography{referencesIT}

\end{document}